\documentclass{article}
\usepackage{amsthm,amsmath,amssymb,amsfonts,xspace,graphicx,multirow}
\usepackage[all]{xy}

\title{Interactive certificate for the verification of {Wiedemann}'s
  {Krylov} sequence: application to the certification of the determinant, the
  minimal and the characteristic polynomials of sparse matrices}

\author{Jean-Guillaume Dumas\thanks{
    Laboratoire J. Kuntzmann,
    Universit\'e de Grenoble. 51, rue des Math\'ematiques, umr CNRS
    5224, bp 53X, F38041 Grenoble, France,
    \href{mailto:Jean-Guillaume.Dumas@imag.fr}{Jean-Guillaume.Dumas@imag.fr},
    \href{http://ljk.imag.fr/membres/Jean-Guillaume.Dumas/}{ljk.imag.fr/membres/Jean-Guillaume.Dumas}.
  }
  \and Erich Kaltofen\thanks{
    Department of Mathematics.
    North Carolina State University.
    Raleigh, NC 27695-8205, USA.
    \href{mailto:kaltofen@math.ncsu.edu}{kaltofen@math.ncsu.edu},
    \href{http://www.kaltofen.us}{www.kaltofen.us}.
  }
  \and Emmanuel Thom\'e\thanks{
    CARAMEL Project -- INRIA Nancy Grand Est.
    615 rue du Jardin Botanique--54602 Villiers-les-Nancy -- France.
    \href{mailto:Emmanuel.Thome@gmail.com}{Emmanuel.Thome@gmail.com},
    \href{http://www.loria.fr/~thome}{www.loria.fr/\~{}thome/}.
  }
}

\newcommand{\Z}{\ensuremath{{\mathbb Z}}}
\newcommand{\F}{\ensuremath{{\mathbb F}}}
\newcommand{\Su}{\ensuremath{{\mathbb S}}}
\newcommand{\Seq}[1]{\ensuremath{W(#1)}}
\newcommand{\bigO}[1]{\ensuremath{{\mathcal O}\left(#1\right)}}
\usepackage{algorithm,algorithmic}

\newcommand{\checks}{\ensuremath{\stackrel{?}{==}}}
\newcommand{\REDUC}[2]{{\sc #1}$\prec${\sc #2}}
\newcommand{\REDUCS}[3]{{\sc #1}$\prec${\sc #2}$\prec${\sc #3}}

\newtheorem{theorem}{Theorem}
\newtheorem{corollary}{Corollary}
\newtheorem{lemma}{Lemma}
\newtheorem{definition}{Definition}

\usepackage{color,svgcolor}
\usepackage[plainpages=true]{hyperref}
\makeatletter
\hypersetup{
pdftitle={\@title},
pdfauthor={Dumas,Kaltofen,Thom\'e},
breaklinks=true,
colorlinks=true,
 linkcolor=darkred,
 citecolor=blue,
 urlcolor=darkgreen,
}
\makeatother

\begin{document}
\maketitle
\begin{abstract}
Certificates to a linear algebra computation are
additional data structures for each output, which can be used
by a---possibly randomized---verification algorithm that proves the
correctness of each output. 
Wiedemann's algorithm projects the Krylov sequence obtained by repeatedly
multiplying a vector by a matrix to obtain a linearly recurrent sequence. 
The minimal polynomial of this sequence divides the minimal polynomial of the
matrix. 
For instance, if the $n\times n$ input matrix is sparse with $n^{1+o(1)}$
non-zero entries, the computation of the sequence is quadratic in the dimension
of the matrix while the computation of the minimal polynomial is $n^{1+o(1)}$,
once that projected Krylov sequence is obtained.

In this paper we give algorithms that compute certificates for the Krylov
sequence of sparse or structured $n\times n$ matrices over an abstract field,
whose Monte Carlo verification complexity can be made essentially linear. 
As an application this gives certificates for the determinant, the minimal and
characteristic polynomials of sparse or structured matrices at the same cost.
\end{abstract}

\section{Introduction}
We consider a square sparse or structured matrix $A\in\F^{n\times{}n}$. 
By sparse or structured we mean that the multiplication of a vector by $A$
requires less operations than that of a dense matrix-vector multiplication. 
The arithmetic cost to apply $A$ is denoted by $\mu$ which thus satisfies
$\mu\leq{}n(2n-1)$ ($n^2$ multiplications and $n(n-1)$ additions). 
In the following we also need to perform row-vector-times-matrix
multiplications, which, by the transposition principle, cost $O(\mu)$
operations~\cite{Bostan:2003:Tellegen}. 
In the following we will simply consider that both operations (left or right
multiplication by a row or column vector) cost less than $\mu$ arithmetic
operations.

The main idea of this paper is to use a Baby-step/Giant-step verification of
Wiedemann's Krylov sequence generation. 
Once the sequence is verified, the remaining operations, of lower cost, can be
replayed by the Verifier.

The verification procedure used throughout this paper is that of {\em essentially
optimal interactive certificates} with the taxonomy
of~\cite{jgd:2014:interactivecert}. 
Indeed, in the following, we consider a {\em Prover}, nicknamed {\em Peggy}, who
will perform a computation, potentially together with additional data
structures. We also consider a {\em Verifier}, nicknamed {\em Victor}, who will
check the validity of the computation, faster that just by recomputing it. 

By {\em certificates} for a problem that is given by input/output
specifications, we mean, as in \cite{KLYZ09,Kaltofen:2011:quadcert}, 
an input-dependent data structure and an algorithm that computes from that input
and its certificate the specified output, and that has lower computational
complexity than any known algorithm that does the same when only receiving the
input. Correctness of the data structure is not assumed but validated by the
algorithm. 

By {\em interactive certificate}, we mean certificates modeled as
$\sum$-protocols (as defined in~\cite{Cramer:1997:Sigma}) were 
the Prover submits 
a {\em Commitment}, that is some result of a computation; 
the Verifier answers by
a {\em Challenge}, usually some uniformly sampled random values;
the Prover then answers with 
a {\em Response}, that the Verifier can use to convince himself of the validity
of the commitment. 
To be useful, such proof systems is said to be {\em complete} if the probability
that a true statement is rejected by the Verifier can be made arbitrarily
small. Similarly, the protocol is {\em sound} if the probability that a false
statement is accepted by the verifier can be made arbitrarily small.
In the following we will actually only consider {\em perfectly complete}
certificates, that is were a true statement is never rejected by the Verifier.

There two may ways to design such certificates.
On the one hand, efficient protocols can be designed for delegating
computational tasks. In recent years, generic protocols have been designed for
circuits with polylogarithmic
depth~\cite{Goldwasser:2008:delegating,Thaler:2013:crypto}.   
The resulting protocols are interactive and their cost for the Verifier is
usually only roughly proportional to the input size. They however can
produce a non negligible overhead for the Prover and are restricted to
certain classes of circuits. 
Variants with an amortized cost for the Verifier can also be designed,
see for instance~\cite{Parno:2013:Pinocchio}, quite often using relatively
costly homomorphic routines.
Moreover, we want the Verifier to run faster than the Prover, so we discard
amortized models where the Verifier is allowed to do a large amount of
precomputations, that can be amortized if, say, the same matrix is repeatedly
used~\cite{Chung:2010:delfhe,Gentry:2014:nizkfhe}.

On the other hand, dedicated certificates (data structures and
algorithms that are verifiable a posteriori, without interaction) have
also been developed in the last few years, e.g., for dense exact linear
algebra~\cite{Freivalds:1979:certif,Kaltofen:2011:quadcert,Fiore:2012:PVD},
even for problems that have no good circuit representation. 
There the certificate constitute a proof of correctness of a result,
not of a computation, and can thus also stand a direct public verification.
The obtained certificates are ad-hoc, but try to reduce as much as possible the
overhead for the Prover, while preserving a fast verification procedure.

In the current paper we follow the later line of research, that is ad-hoc
certificate with fast verification and negligible overhead for the Prover.\\

In exact linear algebra, the most simple problem to have an optimal certificate
is the linear system solution, {\sc LinSolve}: for a matrix $A$ and a vector
$b$, checking that $x$ is actually a solution is done by  one multiplication of
$x$ by $A$. The cost of this check similar to that of just enumerating all the
non-zero coefficients of $A$. Thus certifying a linear system is reduced to
multiplying a matrix by a vector: \REDUC{LinSolve}{MatVecMult}.
In~\cite{jgd:2014:interactivecert}, two essentially optimal reductions have been
made, that the rank can be certified via certificates for linear systems, and
that the characteristic polynomial can be certified via certificates for the
determinant: \REDUC{CharPoly}{Det} and \REDUC{Rank}{LinSolve}.
But no reduction was given for the determinant. 
We bridge this gap in this paper.
We first use Wiedemann's reduction of the determinant
to the minimal polynomial of a sequence,
\REDUCS{Det}{MinPoly}{Sequence},~\cite{Wiedemann:1986:SSLE}, and then
show that the computation of a sequence generated by projections of
matrix-vector iterations can be checked by a small number of matrix-vector
multiplications: \REDUC{Sequence}{MatVecMult}.\\

The complexity model we consider here is the algebraic complexity model: 
we count field operations, but tests (even such as checking the equality of
whole vectors) are free and uniform sampling of random elements in a field is
also free. 
This is justified by the fact that for all our proposed certificates, the number
of equality tests is always lower than that of field operations and that the
number of random samples is always lower than that of the communications, itself
lower than that of the Verifier's work.

The paper is organized as follows. We define Wiedemann's Krylov sequence
formally in Section~\ref{sec:krylov}. Then we use a check-pointing technique to
propose a first non-quadratic certificate in Section~\ref{sec:checkpoint}.
Then we derive from this technique a recursive process that can yield a method
of decreasing complexities for the Verifier in Section~\ref{sec:recursive}. 
The same general idea is modified in Section~\ref{sec:nlogn} to get a certificate verifiable in essentially optimal time. 
Finally, we show in Section~\ref{sec:applis} how to derive certificates for the
determinant, the minimal and the characteristic polynomial from these
certificates for the Krylov sequence.

\section{Wiedemann's Krylov sequence}\label{sec:krylov}
We consider here the simple Wiedemann's sequence $S$ (no blocks), defined for
two given vectors.

\begin{definition} 
  For $A\in\F^{n\times{}n}$, $V_0\in\F^n$ and $U\in\F^n$,
  {\em Wiedemann's Krylov space} is defined for $i\geq 0$ as:
  $$ K_{V_0} = (V_i)_i = (A^i V_0)_i$$
  {\em  Wiedemann's Krylov sequence} is also defined as:
  $$ S=(s[i])_{i} = (U^T A^i V_0)_i = (U^T V_i)_i$$
\end{definition}

In the following, the Prover
will compute this sequence, potentially together with additional data
structures, and the Verifier will
check the validity of the sequence, once computed.

Now, for a matrix $A$ whose matrix-vector multiplication costs $\mu$ arithmetic
operations, the original cost for the computation of $2n$ elements of
Wiedemann's Krylov sequence is trivially:
\[\Seq{n}=2n\mu+4n^2=\bigO{n\mu}.\]

We summarize in table~\ref{tab:recap}, the complexity bounds for 
certificates of Wiedemann's Krylov sequence, presented in this paper .

\newcommand{\PowCertVerif}{\ensuremath{\mu\log_2(n)+6n\log_2^2(n)}}
\newcommand{\PowCertMem}{\ensuremath{2n\log_2^2(n)}}
\newcommand{\PowCertProve}{\ensuremath{7\Seq{n}}}
\newcommand{\LogPowCertVerif}{\ensuremath{\frac{1}{2}\mu\log_2^2(n)+4n\log_2^2(n)}}
\newcommand{\LogPowCertMem}{\ensuremath{\frac{3}{2}n\log_2^2(n)}}
\newcommand{\LogPowCertProve}{\ensuremath{5\Seq{n}}}
\begin{table}[htb]
\[\begin{array}{|c||c|c|c|}
\hline
\multirow{2}{*}{\text{Certificate}} & \multirow{2}{*}{\text{Verifier}} & \text{Extra} & \multirow{2}{*}{\text{Prover}}\\
&& \text{Communication} &\\
\hline
\S~\ref{sec:checkpoint} &\bigO{n\sqrt{\mu}} &\bigO{n\sqrt{\mu}}& \Seq{n}\\
\S~\ref{sec:dense} &2\mu+\bigO{n\sqrt{n}} &\bigO{n\sqrt{n}}& \Seq{n}+\bigO{\mu\sqrt{n}}\\
\S~\ref{sec:twolev} &4\mu+\bigO{n\sqrt[3]{n}} & \bigO{n\sqrt[3]{n}} &
\Seq{n}+\bigO{\mu{}n^{2/3}}\\
\S~\ref{sec:morelev} &2^k\mu+\bigO{n\sqrt[k]{n}} & \bigO{n\sqrt[k]{n}} & \Seq{n}+o(\Seq{n})\\
\S~\ref{sec:nlogn} & \bigO{\mu\log^2(n)} & \bigO{n\log^2(n)}& \LogPowCertProve \\
\S~\ref{sec:nlogn} & \bigO{\mu\log(n)+n\log^2(n)} & \bigO{n\log^2(n)}& \PowCertProve \\
\hline
\end{array}\]
\caption{Summary of the complexity bounds of the certificates presented in this
  paper for Wiedemann's Krylov sequence}\label{tab:recap}
\end{table}

\section{An \texorpdfstring{$n^{1+1/2}$}{n1.5}
  certificate}\label{sec:checkpoint}
\subsection{A four steps Baby-step/Giant-step interactive protocol}
The protocol has four steps: Victor first selects the vectors for the sequence
that are sent to Peggy. Peggy then computes the sequence and keeps some of the
intermediate vectors, called checkpoints. She then sends the sequence to Victor
together with the additional intermediate vectors which Victor will use to
certify the received sequence:
\begin{enumerate}
\item Communications from Victor to Peggy
\begin{enumerate}
\item Uniformly sample $V_0\in\F^n$, $U\in\F^n$;
\item Sends $A$, $U$, $V_0$.
\item Asks for a sequence of $\delta+1$ elements.
\item Asks for a checkpoint every $K< \min\{n,\delta\}$ matrix-vector products.
\end{enumerate}
Communication is $|A|+2n\leq \mu+2n$.
\item Computations of Peggy: 
\begin{enumerate}
\item $V_i = A V_{i-1}$ for $i=0..\delta$;
\item $s[i]=U^T V_i$ for $i=0..\delta$.
\end{enumerate}
Complexity is exactly that of Wiedemann's sequence; 
that is \bigO{n\mu+n^2} if $\delta=2n$.

\item Communications from Peggy to Victor
\begin{enumerate}
\item Sends $W_j=V_{j K} = A^{j K} V_0$ for $j=0..\lceil\frac{\delta}{K}\rceil$;
\item Sends $s[i]$ for $i=0..\delta$.
\end{enumerate}
Communication is $n\lceil\frac{\delta}{K}\rceil+\delta+1=\bigO{\delta\frac{n}{K}}$.

\item Verifications of Victor.
\newcounter{myenumii}\begin{enumerate}
\item Uniformly sample $R=(r[i])\in\F^K$ and $X\in\F^n$, with
$X\neq U$.
\setcounter{myenumii}{\value{enumii}}
\end{enumerate}
Then first compute some baby steps:
\begin{enumerate}\setcounter{enumii}{\themyenumii}
\item\label{computeZ} Compute $Z = X^T A^K$, in $K\mu$ operations;
\item\label{computeT} Compute $T = \sum_{i=0}^{K-1} r[i] U^TA^i $ in $(K-1)\mu+2Kn$ operations.
\setcounter{myenumii}{\theenumi}\end{enumerate}
For each $j=1..\lceil\frac{\delta}{K}\rceil$;
\begin{enumerate}\setcounter{enumii}{\themyenumii}
\item\label{checkZ} $X^T W_j \checks{} Z W_{j-1}$ in $2\times{}2n+n$
  operations; {\hfill// Checks the $W_j$ with $X$}
\item\label{checkT} $\sum_{i=0}^{K-1} r[i] s[jK+i] \checks{} T W_j$; in
  $2K+2n+1$ operations. {\hfill// Checks the $s[i]$ with $R$ once $W_j$ is certified}
\end{enumerate}
\end{enumerate}
The overall complexity of the verification step is bounded by:
\begin{equation}\label{eq:verifcomp}
2K(\mu+n)+\left\lceil\frac{\delta}{K}\right\rceil(2K+6n).
\end{equation}

\begin{lemma}\label{lem:checkZT} The above protocol is perfectly complete.
\end{lemma}
\begin{proof}
\begin{itemize}
\item[\ref{checkZ}:] $X^T W_j = X^T V_{jK}=X^T A^{jK}V_0=X^T A^KA^{(j-1)K}V_0$,
  so that we also have $X^T W_j =X^T A^K
  V_{(j-1)K}= Z^T W_{j-1}$. 
\item [\ref{checkT}:] $r[i] s[jK+i] = r[i] U^T V_{jK+i} = r[i] U^T A^i V_{jK} =
  r[i] U^T A^i W_j$;
\end{itemize}
\end{proof}

\subsection{Optimal Verifier complexity}
\begin{theorem}\label{thm:onepointhalf} 
  Let $A\in\F^{n\times{}n}$ whose matrix-vector product can be computed in less
  than $\mu>n$ arithmetic operations and a vector $V_0\in\F^n$. 
  There exists a
  certificate of size: 
  \[ \frac{1}{\sqrt{3}}\sqrt{\delta n(\mu+n)}\] 
  for the
  $\delta+1$ first elements of Wiedemann's Krylov sequence associated to~$A$
  and~$V_0$. 
  This certificate is verifiable in time: 
  \[ 4\sqrt{3}\sqrt{\delta  n(\mu+n)}.\]
  With $\mu=n^{1+o(1)}$, and $\delta=2n$, this is a Verifier in $n^{1.5+o(1)}$ time
  and communications.
\end{theorem}
\begin{proof}
The optimal value of $K$ minimizes Equation~(\ref{eq:verifcomp}) and is
therefore close to: \[ K \approx \sqrt{3}\sqrt{\frac{n\delta}{\mu+n}} \]
Substituting the latter into Equation~(\ref{eq:verifcomp}) gives the announced
time complexity. For the size of the certificate, apart from the matrix $A$
itself, the additional communications are the initial vectors sent by Victor and
the intermediate check-pointing vectors sent by Peggy. Once again substituting
the value for $K$ gives the announced complexity.
\end{proof}

We ran this choice on a very sparse matrix with $3$ non zero elements per row.
Results are shown in Table~\ref{tab:veriftime}: computing the sequence took two
hours, the thousand $W$ checkpoints required about two giga bytes of data,
and were checked in a little more than half a minute.
\begin{table}[htb]\center
\begin{tabular}{|c|c|cccc|}
\hline
\multirow{2}{*}{Prover} & \multirow{2}{*}{Communications} &
\multicolumn{4}{|c|}{Verifier} \\ 
& & Z & Check Z  & T & Check T\\
\hline
1.8 hours & 1.9 GB & 5.6 s & 14.5 s & 7.0 s & 6.0 s\\
\hline
\end{tabular}
\caption{Verification for a matrix with $m=n=253008$, $759022$ non-zeroes and
  of compressed size of $3.8$MB. This is $506046$ iterations, and $K=503$ was chosen
  on one core of an i5-4690
  \href{}{@3.50GHz}}\label{tab:veriftime}
\end{table}

\subsection{Soundness}
For the soundness, we need to sample from a finite subset $\Su$ of $\F$.
\begin{theorem}\label{thm:soundness} If the Verifier samples $R$ and $X$ uniformly and independently
  from a finite subset $\Su\subseteq\F$, then the Verifier mistakenly misses any error
  in the sequence or in the check-pointing vectors with probability $\le
  1/|\Su|$. 
\end{theorem}
\begin{proof}
\begin{enumerate}
\item $W_0=V_0$ is given. 
  Thus, inductively, Peggy must find $W_j$ for each $j\geq 1$ such that
  $M_j=W_j-A^KW_{j-1}$ satisfies $X^T M_j =0$, for a random
  secret~$X$ unknown to her. If $M_j$ is non zero then there is $1/|\Su|$
  chances that the dot-product is zero. 
\item Afterwards, 
  let $\Theta_j$ be the vector of $\Theta_j[i] = U^T A^i W_j = U^T A^{jK+i} V_0$.
  $W_j$ being correct, Peggy must find a vector $\Delta_j$ with 
  $\Delta_j[i] = s[jK+i]$ such that $N_j=\Delta_j-\Theta_j$ satisfies 
  $R^T N_j=0$, for a random secret~$R$ unknown to her. If $N_j$ is non
  zero there is $1/|\Su|$ chances that the dot-product is zero.
\end{enumerate}
\end{proof}
To improve probability, as usual, it is possible to rerun the protocol with some
other vectors $X$ and $R$, $\ldots$

\subsection{Public verifiability}
The protocol is publicly verifiable. Indeed, no response from the Prover is
requested after the selection of the challenge $X$ and $R$. 
Therefore, any external participant can also generate its own $X$ and $R$ 
and re-check the Krylov space vectors and Wiedemann's sequence, at the cost
given in Theorem~\ref{thm:onepointhalf}.

\subsection{Constants for block Wiedemann's algorithm}\label{sec:blockwied}
It is possible to use the same protocol to check the matrix sequence produced in
the block Wiedemann's algorithm~\cite{Coppersmith:1994:SHL} with a projection of
$s_1$ vectors on the left and $s_2$ vectors on the right.
The following modifications have to be made, mainly replacing some vectors by
blocks of vectors:
\begin{itemize}
\item $U\in\F^{n\times{}s_1}$, $V_i\in\F^{n\times s_2}$, $W_j\in\F^{n\times s_2}$,
  $S[i]\in\F^{s_1\times{}s_2}$;
\item $X$ and $Z$ remain in $\F^n$ while $R\in\F^{K\times{}s_1}$
  and $r[i]$ is in fact the transpose of a vector in $\F^{s_1}$;
\item and $T\in\F^{s_1\times{}n}$.
\end{itemize}
The length of the sequence is now
$\ell=\frac{n}{s_1}+\frac{n}{s_2}+\bigO{1}$~\cite{Kaltofen:1995:ACB,Villard:1997:further}.
\begin{enumerate}
\item Communications become: 
$\lceil\frac{\ell}{K}\rceil(ns_2)+\ell s_1 s_2$.
\item Verifications become:
$(K\mu)+K(s_1\mu+2 s_1 n+n)+\lceil\frac{\ell}{K}\rceil(4ns_2+K(2s_1s_2+s_2)+2ns_1s_2)$
\end{enumerate}
Now the optimal $K$ becomes:
\[ K \approx
\sqrt{1+\frac{2}{s_1}}\sqrt{\frac{s_2}{s_1}}\sqrt{\frac{\ell{}n}{\mu{}(\frac{1}{2}+\frac{1}{2s_1})+n}}
\]
As $(\frac{1}{2}+\frac{1}{2s_1})\leq 1$, this is a Verifier in time bounded by:
\[
2\sqrt{s_2}\sqrt{s_1+2}(s_1+1)\sqrt{\ell{}n(\mu+n)}+2\ell{}s_1s_2
\]
With $\mu=n^{1+o(1)}$ and $s_1=s_2=s$, the length of the sequence is
$\ell\approx{}2\frac{n}{s}$ so that the Verifier time becomes $(sn)^{1+1/2+o(1)}$.


\section{Recursive verification}\label{sec:recursive}
In fact, in the verification steps of Victor, in the protocol of
Section~\ref{sec:checkpoint}, it is possible to also delegate the computation of $Z$
and $T$. 
\subsection{Denser matrices, Verifier in time
\texorpdfstring{$2\mu + n^{1+1/2+o(1)}$}{2mu+n1.5}}\label{sec:dense}

Next, we propose to delegate just the matrix-vector operations, so that we get a
good complexity also on matrices with more than $n^{1+o(1)}$ entries. 
The idea is that the Verifier can delegate his computations of several
successive matrix vector product and check the whole list of computed vectors. 
Therefore he replaces matrix-vector products by checks of validity of
vectors. The trick is that verifying a vector can be done with a single
dot-product of cost $2n$, while multiplying a matrix by a vector costs $\mu$ 
operations. 

This way, correctness of a full Krylov space can be checked as given in
Algorithm~\ref{alg:kspace}. 

\begin{algorithm}[htb]
\caption{Checking the Krylov Space}\label{alg:kspace}
\begin{algorithmic}[1]
\REQUIRE a matrix $A\in\F^{n\times{}n}$ and a vector $V_0\in\F^n$;
\REQUIRE a list of $d$ vectors $[V_0,V_1,\ldots,V_{d-1}]$;
\ENSURE  $[V_0,V_1,\ldots,V_{d-1}]=[V_0,AV_0,\ldots,A^{d-1}V_0]$.
\STATE For $\Su\subseteq\F$, uniformly sample $Y\in\Su^n$;
\STATE Compute $H=Y^T A$;
\RETURN $H V_{i-1} \checks Y^T V_i $, for $i=1..d$.
\end{algorithmic}
\end{algorithm}

\begin{lemma} Algorithm~\ref{alg:kspace} is sound, perfectly complete and
  requires $\mu+4dn$ arithmetic operations.
\end{lemma}
\begin{proof}
Perfect completeness is granted inductively because $V_0$ is known and then
since $H V_{i-1} = Y^T A A^{i-1}V_0 = Y^TA^i V_0=Y^TV_i$.
Soundness is granted because whenever $AV_{i-1}-Vi\neq 0$, its dot-product with
a uniformly selected $Y\in\Su^n$ will be zero only with probability
$|\Su|^{-1}$. 
Complexity for the Verifier is $\mu$ operations to compute $H$ and then $d$
checks performed by two dot-products of size $n$.
\end{proof}

The idea is to delegate the computation of both $Z$ (Point~\ref{computeZ} of the
protocol of section~\ref{sec:checkpoint}) and $T$ (Point~\ref{computeT} of the
protocol of section~\ref{sec:checkpoint}). Then to only check both resulting
Krylov spaces.
Note that it is mandatory that this delegation of the computation of $Z$ and 
$T$ takes place {\em after} the commitment of the $W_j$ and the $s[i]$ by the
Prover.

In the complexity of Theorem~\ref{thm:onepointhalf}, this
replaces  two $K\mu$ factors (now an additional, but neglectible, cost to the
Prover), each by a $\mu+4Kn$ factor. This gives a new complexity of
$2\mu+10Kn+\left\lceil\frac{\delta}{K}\right\rceil(2K+6n)$ for the Verifier.
There are some extra communications, the vectors used for the computation of $Z$
and $T$, namely $2n(K-1)$ field elements.
We have proven:
\begin{corollary}\label{cor:dense}
  Let $A\in\F^{n\times{}n}$ whose matrix-vector product can be computed in less
  than $\mu>n$ arithmetic operations and a vector $V_0\in\F^n$. 
For any $1\leq K\leq min\{n,\delta\}$, there exists a sound and perfectly complete
protocol verifying the first $\delta+1$ elements of Wiedemann's Krylov sequence
associated to $A$ and~$V_0$, in time
$2\mu+10Kn+\left\lceil\frac{\delta}{K}\right\rceil(2K+6n)$. The associated
certificate has size $n\left\lceil\frac{\delta}{K}\right\rceil+2nK$.
\end{corollary}
The extra work for the Prover is that of the computation of $Z$ and $T$, both
bounded by $\bigO{\mu K}=\bigO{\mu n^{2/3}}$, negligible with respect to the
computation of the sequence, $\bigO{\mu n}$. 

In terms of computational time for the verifier, the associated optimal $K$
factor becomes $K=\sqrt{\frac{3}{5}}\sqrt{\delta}$ and the Verifier complexity
is transformed into: \[4n+2\mu+4n\sqrt{15\delta}.\] 
With $\delta=2n$, this gives a Verifier complexity bounded by
$2\mu+21.91n^{1.5}$, with a certificate of size bounded by $4.02n^{1.5}$.

\subsection{Optimal 2-levels of recursion and an \texorpdfstring{$n^{1+1/3}$}{n1.333} certificate for Wiedemann's algorithm}\label{sec:twolev}
Now, instead of just delegating the matrix-vector products, we
delegate the whole computation of $Z$ and~$T$:
\begin{enumerate}
\item For $Z$, it is actually sufficient to reuse the scheme of
  Section~\ref{sec:dense} with $\delta=K$, choosing a $K_2<K$, and $Z$ will be
  certified as the last $W_j$ vector.
  The time for the Verifier for this step is thus bounded by
  $2\mu+10nK_2+\frac{K}{K_2}(2K_2+6n)$. 
\item For $T$, the protocol is twofold:
\begin{enumerate}
\item Send the $r[i]$, $U$ and $A$, and ask just for $T$ in return;
\item Only now, send a uniformly sampled vector $\Psi$ and ask for a certificate of the
  sequence $\Gamma=\gamma[i]=U^TA^i\Psi$;
\item Then one can check that $\sum r[i] \gamma[i] \checks{} T \Psi$.
\end{enumerate}
\end{enumerate}

\begin{theorem}\label{thm:dense}
  For $A\in\F^{n\times{}n}$ whose matrix-vector product can be computed in less
  than $\mu>n$ arithmetic operations and a vector $V_0\in\F^n$,
  there exists a sound and perfectly complete interactive
  certificate for the associated Wiedemann's Krylov sequence of size
  $\bigO{n^{1+1/3}}$. 
  This certificate is verifiable in time
  \[ 4\mu+\bigO{n^{1+1/3}}. \]
\end{theorem}
\begin{proof}
We still use the protocol of Section~\ref{sec:checkpoint}, but replace the
computation of $Z$ and $T$ by the above delegated scheme.

The protocol is perfectly complete, since 
$\sum r[i]\gamma[i] = \sum r[i] (U^T A^i \Psi) =  \sum (r[i] U^T A^i) \Psi
= T \Psi$.

The protocol is sound because the $\gamma[i]$ are correctly verified by a
sound protocol. Then $\Psi$ being unknown when asking for
$T$, $T$ cannot be engineered to satisfy the last check: 
if $G=T-\sum r[i](U^TA^i)$ is non zero then there is $1/|\Su|$ chances that its
dot-product with $\Psi$ is zero.

Verifier time and space for $T$ are that of Corollary~\ref{cor:dense} for the
sequence $\Gamma$, and a supplementary dot-product. Verifier time and space for
$Z$ are also that of Corollary~\ref{cor:dense}.
Therefore, since $6n+2K\leq 8n$, overall, the Verifier runs now in time bounded
by: 
\begin{equation}\label{eq:twolev}
2\left(2\mu+10nK_2+8n\frac{K}{K_2}\right)+2n+8n\left\lceil\frac{\delta}{K}\right\rceil=4\mu+\bigO{nK_2+n\frac{K}{K_2}+n\frac{\delta}{K}}
\end{equation} 
with a certificate of size bounded by:
\[
\bigO{n\frac{\delta}{K}+n\frac{K}{K_2}+nK_2}.
\]

With $\delta=2n$, optimal values for $K$ and $K_2$ are now
respectively $n^{2/3}$ and $n^{1/3}$ for a Verifier in time $4\mu
+ \bigO{n^{1+1/3}}$ with a certificate of size $\bigO{n^{1+1/3}}$.  
\end{proof}

The extra work for the Prover is that of the computation of $Z$ and $T$ 
both bounded by $\bigO{\mu K}=\bigO{\mu n^{2/3}}$, of $\Gamma$ (if done together
with that of $T$, this requires only $K$ dot-products), and of the $Zs$ and $Ts$
for the verifications of $Z$, $T$ and $\Gamma$. Those are bounded by $\bigO{\mu
  K_2}=\bigO{\mu n^{1/3}}$. All this is negligible with respect to the
computation of the sequence, $\bigO{\mu n}$.

\subsection{More levels and a Verifier in time \texorpdfstring{$n^{1+1/k+o(1)}$}{n{1+1/k}}}\label{sec:morelev}
Once it is proven that the computation of $Z$ and $T$ can be delegated, then the
computation of $Z_2$ and $T_2$ in their verification can also be delegated.
The idea, is thus to use the protocol of section~\ref{sec:twolev}, also for $Z$
and $T$, but with two parameters $K_1$ and $K_2$ to
set and $\delta=K$ in equation~(\ref{eq:twolev}). 
The verification time for $Z$ and $T$ becomes
$4\mu+\bigO{nK_2+n\frac{K_1}{K_2}+n\frac{K}{K_1}}$ for each and, overall, the Verifier
thus runs now in time bounded by:
\begin{equation}\label{eq:threelev}
8\mu+\bigO{\left(nK_2+n\frac{K_1}{K_2}+n\frac{K}{K_1}\right)+n\frac{\delta}{K}}.
\end{equation}
With $K_2=n^{\alpha_2}$, $K_1=n^{\alpha_1}$, $K=n^\beta$, the optimal
values should equal
$1+\alpha_2=1+\alpha_1-\alpha_2=1+\beta-\alpha_1=2-\beta$, or differently
written,
$2\alpha_2-\alpha_1=0;-\alpha_2+2\alpha_1-\beta=0;-\alpha_1+2\beta=1$.
This yields $[\alpha_2=1/4,\alpha_1=1/2,\beta=3/4]$, so that 
$K_2=n^{1/4}$, $K_1=n^{1/2}$, $K=n^{3/4}$ and the complexity is bounded:
$$8\mu+\bigO{n^{1+1/4}}.$$
As previously, the size of the certificate is also reduced to $\bigO{n^{1+1/4}}$
and the extra work for the Prover is increased to $\bigO{\mu{}n^{3/4}}$, still
negligible with respect to $\bigO{\mu{}n}$.\\

More generally, for any $k$, we have 
\begin{equation}
\left[\begin{array}{ccccccc}
2 & -1 & 0 & 0 & \ldots & 0 \\
-1 & 2 & -1 & 0 & \ldots & 0\\
0 & \ddots & \ddots & \ddots & \ddots & \vdots \\
\vdots & \ddots & \ddots & \ddots & \ddots & 0\\
0 & \ldots & 0 & -1 & 2 & -1 \\
0 & \ldots &  \ldots & 0 & -1 & 2\\
\end{array} \right]
\left[\begin{array}{c}
\alpha_{k-2}\\
\alpha_{k-3}\\
\vdots\\
\alpha_2\\
\alpha_1\\
\beta\\
\end{array} \right]
=
\left[\begin{array}{c}
0\\
\vdots\\
\vdots\\
\vdots\\
0\\
1\\
\end{array} \right]
\end{equation}
For $L$ a unit lower triangular matrix, the latter gives, via Gaussian
elimination without pivoting:
\begin{equation}
\left[\begin{array}{ccccccc}
2 & -1 & 0 & 0 & \ldots & 0 \\
0 & \frac{3}{2} & -1 & 0 & \ldots & 0\\
0 & \ddots & \frac{4}{3} & \ddots & \ddots & \vdots \\
\vdots & \ddots & \ddots & \ddots & \ddots & 0\\
0 & \ldots & 0 & 0 &  \frac{k-1}{k-2} & -1 \\
0 & \ldots &  \ldots & 0 & 0 & \frac{k}{k-1}\\
\end{array} \right]
\left[\begin{array}{c}
\alpha_{k-2}\\
\alpha_{k-3}\\
\vdots\\
\alpha_2\\
\alpha_1\\
\beta\\
\end{array} \right]
= L^{-1}
\left[\begin{array}{c}
0\\
\vdots\\
\vdots\\
\vdots\\
0\\
1\\
\end{array} \right]
=
\left[\begin{array}{c}
0\\
\vdots\\
\vdots\\
\vdots\\
0\\
1\\
\end{array} \right]
\end{equation}
So that the solution is:
\begin{equation}
\left[\begin{array}{cccccc}
\alpha_{k-2}&
\alpha_{k-3}&
\ldots&
\alpha_2&
\alpha_1&
\beta
\end{array} \right]
= 
\left[\begin{array}{ccccccc}
\frac{1}{k}&
\frac{2}{k}&
\ldots&
\frac{k-2}{k}&
\frac{k-1}{k}
\end{array} \right]
\end{equation}
Thus
$n^{1+\alpha_{k-2}}=n^{1+\alpha_{k-3}-\alpha_{k-2}}=\ldots=n^{1+\beta-\alpha_1}=n^{2-\beta}=n^{1+1/k}$.

The size of the certificate is thus $\bigO{n^{1+1/k}}$, 
the time for the Verifier is $2^k\mu+\bigO{n^{1+1/k}}$ and the extra work for
the Prover becomes 
$\sum_{t=1}^k 2^t \mu n^{1-t/k}=\bigO{\mu{}n^{1-1/k}}$, still negligible with $\bigO{\mu{}n}$.

\newcommand{\PowerCertificate}[1]{\texttt{PowerCertificate\ensuremath{(#1)}}}
\newcommand{\SequenceCertificate}[1]{\texttt{SequenceCertificate\ensuremath{(#1)}}}
\newcommand{\CombinationCertificate}[1]{\texttt{CombinationCertificate\ensuremath{(#1)}}}
\section{\texorpdfstring{$\mu\log(n)+n\log^2(n)$}{Essentially linear} certificate}\label{sec:nlogn}
The same idea actually gives rise to a certificate verifiable with only
$\log_2(n)$ matrix-vector products: use a recursive certificate with
$K=\delta/2$.

We first need to separate the interactive protocol of Section~\ref{sec:twolev}
into atomic parts: a recursive interactive protocol for certifying a single
vector corresponding to a large power of $A$ times an initial vector and a
combination of mutually recursive protocols for the sequence.

\subsection{Certificate for the large powers}\label{sec:power}
We want here to certify that $Z \checks{} A^dV$. For this we will need to check
successive powers of two.

\subsubsection{Certificate for the large powers with a logarithmic number of
  matrix-vector products}\label{sec:logpower}
We define the certificate $\PowerCertificate{A,V,d}$ to be two vectors
$Z,Z_{/2}$ that satisfy $Z \checks{} A^d V$ and 
$Z_{/2}\checks{} A^{\lfloor{}d/2\rfloor}V$. 
Then checking this certificate is shown in
algorithm~\ref{alg:pc}.

\begin{algorithm}[htb]
\caption{Logarithmic Interactive recursive check of \PowerCertificate{A,V,d}}\label{alg:pc}
\begin{algorithmic}[1]
  \REQUIRE Matrix $A\in\F^{n\times{}n}$, vector $V\in\F^n$, exponent $d$;
\REQUIRE A pair of vectors $Z,Z_{/2}= \PowerCertificate{A,V,d}$.
\ENSURE $Z \checks{} A^d V$ \AND $Z_{/2}\checks{} A^{\lfloor{}d/2\rfloor}V$.
\IF{$d==1$}
\RETURN $Z_{/2}\checks{} V$ \AND $Z \checks{} AV$.
\ELSE
\STATE Uniformly sample $W\in\F^n$;
\STATE Request $(Y,Y_{/2})=\PowerCertificate{A^T,W,\lfloor{}d/2\rfloor}$ and
recursively check it;
\IF{$d$ is even}
\RETURN $W^T Z_{/2} \checks{} Y^T V$  \AND $W^T Z  \checks{} Y^T Z_{/2}$
\ELSE
\RETURN $W^T Z_{/2} \checks{} Y^T V$  \AND $W^T Z  \checks{} Y^T (A Z_{/2})$.
\ENDIF
\ENDIF
\end{algorithmic}
\end{algorithm}

\begin{lemma}\label{lem:pc} Algorithm~\ref{alg:pc} is sound and perfectly complete. 
It requires $\log_2(d)$ rounds, $3\log_2(d)n$ communications, $2d\mu$ arithmetic
operations for the Prover, and less than $(\mu+8n)\log_2(d)+\mu$ arithmetic
operations for the Verifier.
\end{lemma}
\begin{proof}
The protocol is perfectly complete by induction: the basis of the induction is 
given by the case $d==1$; then by induction $Y=(A^T)^{\lfloor{}d/2\rfloor}W$, so
that: 
\[ W^T Z_{/2} = W^T A^{\lfloor{}d/2\rfloor} V=(W^TA^{\lfloor{}d/2\rfloor}) V =
Y^T V\]
and, if $d$ is even:
\[ W^T Z = W^T A^d V=   (W^T A^{\lfloor{}d/2\rfloor})
(A^{\lfloor{}d/2\rfloor} V) = Y^T Z_{/2}\]

or, if $d$ is odd:
\[W^T Z = W^T A^d V=   (W^T A^{\lfloor{}d/2\rfloor}) A
(A^{\lfloor{}d/2\rfloor} V) = Y^T A Z_{/2}.\]

The protocol is sound: 
the Prover produces the commitments $Z$ and $Z_{/2}$, then the Verifier sends a
challenge $W$ and the Prover responds with $Y$. 
There, the Prover has two possibilities, either he returns a correct $Y$ or not.
In the first case, as $W$ was chosen uniformly at random,
there are two sub case, either $Z_{/2}$ is wrong or not.
if $Z_{/2}$ was incorrectly chosen so that $Z_{/2}-A^{\lfloor{}d/2\rfloor} V$ is non zero,
there is $1/|\Su|$ chances that its dot-product with $W^T$ is zero and thus that
it can pass the $W^T Z_{/2} \checks{} Y^T V$ check.
Conversely, if $Z_{/2}$ is correct, if $Z$ was incorrectly chosen so that 
$Z-A^{\lceil{}d/2\rceil}Z_{/2}$ is non zero, there is $1/|\Su|$ chances that its
dot-product with $W^T$ is zero. Both tests are not independent but overall there
are less than $1/|\Su|$ chances to pass both of them. 
In the second case $Y$, $Y$ is incorrect but can very well be made to make both
latter dot-products zero, for any values of $Z$, $Z_{/2}$ and $W$.  
But if $Y$ is incorrect, it will not pass the recursive test if
$\lfloor{}d/2\rfloor=1$, and will pass it only with probability $|\Su|^{-1}$ for
other values of $d$. Therefore, if Peggy's commitment was incorrect, the
probability that it passes all the subsequent tests of Algorithm~\ref{alg:pc} is
less than $|\Su|^{-1}$. 

Now, Communication is that of the certificate, the $3$ vectors $W$, $Y$ and
$Y_{/2}$, per recursive call, that is $3\log_2(d)n$. 
Time complexity for the Verifier satisfies $\{T(d)\leq T(d/2)+2*4n+\mu,T(1)=\mu\}$,
that is less than $(\mu+8n)\log_2(d)+\mu$. 
Now the cost has been transferred to the prover, who has to compute the sequence
plus half a sequence, plus a fourth of a sequence, ..., recursively the
overall cost for the Prover is doubled to $2d\mu$.
\end{proof}

\subsubsection{Public verifiability of the large power}
Another view of the verification of Algorithm~\ref{alg:pc} can be given as an
interactive certificate in Figure~\ref{fig:pc}.  
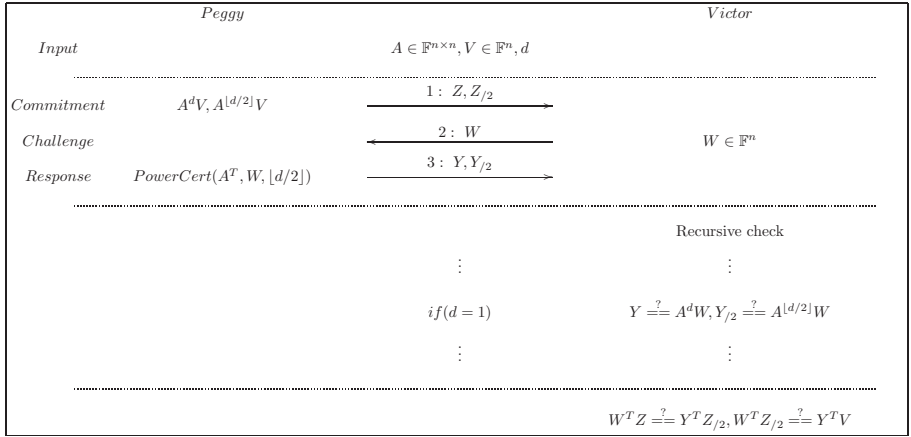
\begin{figure}[htbp]\center
  \noindent\resizebox{\linewidth}{!}{$$
    \xy
    \xymatrix@R=7pt@C=1pc@W=15pt{
      &Peggy &&&& Victor\\
      Input&&& A\in\F^{n\times n}, V\in\F^n, d &&\\
      \ar@{.}[rrrrrr]&&&&&&\\
      Commitment&A^dV, A^{\lfloor{}d/2\rfloor}V&\ar[rr]^*[*1.]{1:\ Z,Z_{/2}}
      &&&  \\
      Challenge& & && \ar[ll]_*[*1.]{2:\ W}& W\in\F^n\\
      Response&PowerCert(A^T,W,\lfloor{}d/2\rfloor) &\ar[rr]^*[*1.]{3:\
        Y,Y_{/2}}&&&  \\
      \ar@{.}[rrrrrr]&&&&&&\\
      &&& && \text{Recursive check} \\
      && &\vdots&& \vdots\\
&& & if(d=1) && Y \checks{} A^d W, Y_{/2} \checks{} A^{\lfloor{}d/2\rfloor} W\\
      && &\vdots&& \vdots\\
      \ar@{.}[rrrrrr]&&&&&&\\
      && &&& W^T Z  \checks{} Y^T Z_{/2},W^T Z_{/2} \checks{} Y^T V\\
    }\POS*\frm{-}
    \endxy
    $$
  }
  \caption{Interactive certificate for $A^dV$}\label{fig:pc}
\end{figure}

As the challenge is only random samples selected after the commitment (and this
is true also recursively), Fiat-Shamir heuristic can be used at each
step~\cite{Fiat:1986:Shamir,Bellare:1993:randomoracle,Bernhard:2012:fiatshamir}:
$W$ can be just the result of a cryptographically strong hash function on $A$,
$V$, $d$, {\em and} $Z$, $Z_{/2}$. 
Then any external verifier can simulate the whole protocol by recomputing also
the hashes.

\subsubsection{Certificate for the large powers with a single matrix-vector product}\label{sec:singlepower}
Actually, algorithm~\ref{alg:pc} can be made to require a {\em single}
matrix-vector product. The speed up for the verifier is obtained by recursively
asking for a little more: some arithmetic cost for the Verifier is traded-off
with an  extra cost for the Prover and some extra communications. 

The certificate $\PowerCertificate{A,V,d}$ is modified to be {\em three}
vectors: for any $t$ such that $2^{t}\geq d$, we
check $A^d V$, together with $A^{2^{t}}V$ and $A^{2^{t-1}}V$. 

\begin{algorithm}[htb]
\caption{Interactive recursive check of \PowerCertificate{A,V,d,2^t}}\label{alg:single}
\begin{algorithmic}[1]
  \REQUIRE Matrix $A\in\F^{n\times{}n}$, vector $V\in\F^n$, exponent $d\geq 2$, $t$
  such that $2^t\geq d$;
\REQUIRE A triple of vectors $Z_{t},Z,Z_{t-1}= \PowerCertificate{A,V,d,2^t}$.
\ENSURE  $Z_{t-1}\checks{} A^{2^{t-1}}V$\AND $Z \checks{} A^d V$ \AND $Z_{t}\checks{} A^{2^{t}}V$ .
\STATE Uniformly sample $W\in\F^n$;
\IF{$d==2$}
\STATE Compute $Y=A^T W$;
\RETURN $W^T Z_0 \checks{} Y^T V$ \AND $W^T Z  \checks{} Y^T Z_0$ \AND
$Z_1\checks{} Z$.
\ELSE
\STATE Request $(Y_{t-1},Y,Y_{t-2})=\PowerCertificate{A^T,W,d-2^{t-1},2^{t-1}}$ and
recursively check it;
\RETURN $W^T Z_{t-1} \checks{} Y^T V$  \AND $W^T Z  \checks{} Y^T Z_{t-1}$ \AND
$W^TZ_t \checks Y_{t-1}^T Z_{t-1}$.
\ENDIF
\end{algorithmic}
\end{algorithm}

\begin{lemma}\label{lem:pc2} Algorithm~\ref{alg:single} is sound and perfectly complete. 
It requires $\log_2(d)$ rounds, $4\log_2(d)n$ communications, $2^{t+1}\mu$, less
than $4d\mu$ arithmetic operations for the Prover, and less than
$\mu+8n+12n\log_2(d)$ arithmetic 
operations for the Verifier.
\end{lemma}

The proof is similar to that of Lemma~\ref{lem:pc}.

\subsection{Certificate for the sequence}
Now the idea is to use the protocol of Section~\ref{sec:checkpoint}, with
$K=\delta/2$, but with the computations of $Z$ and $T$ completely delegated.
The computation of $Z$ can be verified, using either one of the
\PowerCertificate{\ldots} protocols of Section~\ref{sec:power}.
Wiedemann's Krylov sequence and $T$ will then be verified with two distinct
protocols, mutually recursive:
\begin{itemize}
\item For the sequence, with $K=\delta/2$, the verification loop of
  point~\ref{checkZ} is reduced to the verification of two checkpoint vectors 
  $(W,W_{/2})$ and of two parts of the sequence $S=(s[i])=(s_H,s_L)$. 
  Thus the data structure \SequenceCertificate{U,A,V,d} is a combination of two
  vectors $(W,W_{/2})$, a sequence $S=(s[i])$ and two other certificates, one
  for $Z$: \PowerCertificate{A^T,X,d/2} and the second one for the linear
  combination $T$:
  \CombinationCertificate{R,U,A,d/2}, for uniformly sampled $X$ and $R$. 
  The checkpoint vectors satisfy $W \checks{} A^d V$ and
  $W_{/2}\checks{} A^{\rfloor{}d/2\lfloor{}}V$, and the output sequence
  satisfies the expected $S=(s[i])=(s_H,s_L) \checks{} U^T A^i  V$ for $i=0..d$.
\item For the delegation of $T$, it is sufficient to generate a certified
  sequence with another right projection. Thus,
  \CombinationCertificate{R,U,A,d} is a combination of the vector $T$, that must
  satisfy as expected $T \checks{} \sum r[i] U^T A^i$ and of another
  certificate, \SequenceCertificate{U,A,\Psi,d}, for a uniformly sampled
  $\Psi$. 
\end{itemize}

Checking these two certificates is done by using the following two mutually
recursive procedures, shown in algorithms~\ref{alg:sc} and~\ref{alg:cc}.
\begin{algorithm}[htb]
\caption{Interactive check of \SequenceCertificate{U,A,V,d}}\label{alg:sc}
\begin{algorithmic}[1]
  \REQUIRE Matrix $A\in\F^{n\times{}n}$, two vectors $U,V\in\F^n$, sequence length
  $d+1$ with $d\geq 2$; 
\REQUIRE A pair of vectors $W,W_{/2}\in\F^n$;
\REQUIRE A sequence $(s[i])\in\F^{d+1}$.
\ENSURE $W=A^{2\left\lceil\frac{d}{2}\right\rceil}V$ \AND $W_{/2}=A^{\left\lceil\frac{d}{2}\right\rceil}V$;
\ENSURE $s[i] \checks{} U^T A^i V$ for $i=0..d$.
\IF{d==2} 
\RETURN $s[0] \checks{} U^T V$ \AND $W_{/2}\checks{}AV$ \AND $s[1]\checks{}U^TW_{/2}$
\AND $W\checks AW_{/2}$ \AND $s[2]\checks{}U^TW$.
\ELSE
\STATE Uniformly sample $X\in\F^n$;
\STATE Ask for $(Z,\ldots)=\PowerCertificate{A^T,X,\lceil{}d/2\rceil}$ and check it;
\STATE Let $first\leftarrow X^TW_{/2} \checks{} Z^T V$;
\STATE Let $second\leftarrow X^T W \checks{} Z^T W_{/2}$;
\STATE Uniformly sample $R\in\F^{\left\lceil\frac{d}{2}\right\rceil+1}$;
\STATE Ask for $(T,\ldots)=\CombinationCertificate{R,U,A,\left\lceil\frac{d}{2}\right\rceil}$ and check it;
\STATE Let $s_L=(s[0],\ldots,s\left[\left\lceil\frac{d}{2}\right\rceil\right])$ and $third\leftarrow R^T s_L \checks{} T^T W_{/2}$;
\STATE Let $s_H=(s\left[\left\lfloor\frac{d}{2}\right\rfloor\right],\ldots,s[d])$ and $fourth\leftarrow R^T s_H \checks{} T^T W$;
\RETURN $first$ \AND $second $ \AND $third$ \AND $fourth$.
\ENDIF
\end{algorithmic}
\end{algorithm}

\begin{algorithm}[htb]
\caption{Interactive check of \CombinationCertificate{R,U,A,d}}\label{alg:cc}
\begin{algorithmic}[1]
\REQUIRE Matrix $A\in\F^{n\times{}n}$, two vectors $R\in\F^{d+1}$ and $U\in\F^n$, sequence length $d+1$;
\REQUIRE A vector $T\in\F^n$.
\ENSURE $T \checks\sum_{i=0}^d r[i] U^T A^i$.
\STATE Uniformly sample $\Psi\in\F^n$;
\STATE Ask for $(\Gamma,\ldots) =\SequenceCertificate{U,A,\Psi,d}$ and check it;
\RETURN $R^T \Gamma \checks{} T \Psi$.
\end{algorithmic}
\end{algorithm}


\begin{theorem} 
  Let $A\in\F^{n\times{}n}$ whose matrix-vector product can be computed in less
  than $\mu>n$ arithmetic operations and a vector $V_0\in\F^n$. 
  There exists a
  certificate of size $\bigO{n\log(n)}$ for the $\delta+1$ first elements of
  Wiedemann's Krylov sequence associated to~$A$ and~$V_0$. 
  This certificate can be checked using the protocol of Algorithm~\ref{alg:sc}. 
  Depending on the \PowerCertificate{\ldots} routine chosen, the constant factor of this
  size and the Prover and Verifier arithmetic complexity bounds for this
  protocol are given in table~\ref{tab:sc}.
\begin{table}[htbp]
\[\begin{array}{|c||c|c|c|}
\hline
\text{Power} & \multirow{2}{*}{\text{Verifier}} & \text{Extra} &  \multirow{2}{*}{\text{Prover}}\\
\text{Certificate} & & \text{Communication} &\\
\hline
\S~\ref{sec:logpower} & \LogPowCertVerif & \LogPowCertMem & \LogPowCertProve \\
\S~\ref{sec:singlepower} & \PowCertVerif & \PowCertMem & \PowCertProve \\
\hline
\end{array}\]
\caption{Dominant terms of the complexity bounds for the verification of Wiedemann's Krylov sequence
  depending on the certification of $Z\checks{}A^dV$.}\label{tab:sc}
\end{table}

\end{theorem}
\begin{proof}
The protocol is sound and perfectly complete by induction on the size of
sequence: the case $d==1$ in Algorithm~\ref{alg:sc} gives the base of the
induction; then the four explicit checks are correct thanks to
Lemma~\ref{lem:checkZT} and sound thanks to Theorem~\ref{thm:soundness}; 
\PowerCertificate{\ldots} is correct and sound by Lemma~\ref{lem:pc} or
Lemma~\ref{lem:pc2}; and \CombinationCertificate{\ldots} is correct and sound,
first by induction on \SequenceCertificate{\ldots} with half the initial size, and
second, since the explicit check is correct and sound by
Theorem~\ref{thm:dense}.

Complexity for the Verifier of the \SequenceCertificate{\ldots} sequence
satisfies 
\begin{equation*}
\begin{split}
\{ \SequenceCertificate{d} & = \PowerCertificate{d/2} \\
& \quad + \CombinationCertificate{d/2} + 12n+2d,\\
\SequenceCertificate{2}31
&=2\mu+6n\}.
\end{split}
\end{equation*} 

Complexity for the Verifier of $T$ satisfies
$\{\CombinationCertificate{x}=\SequenceCertificate{x}+2n+2x\}$.\\

With $\PowerCertificate{x}=(\mu+8n)\log_2(x)+\mu$ (see Lemma~\ref{lem:pc}), the dominant terms of
the complexity bound for the Verifier is thus: 
\[
 \SequenceCertificate{d}=\frac{1}{2}\mu\log_2^2(d)+4n\log_2^2(d)
\]
Similarly, with $\PowerCertificate{x}=\mu+8n+12n\log_2(x)+\mu$ (see
Lemma~\ref{lem:pc2}) we get:
\[
 \SequenceCertificate{d}=\mu\log_2(d)+6n\log_2^2(d)
\]
With $d=2n$ we obtain the Verifier column of Table~\ref{tab:sc}.

Similarly, communication is dominated either by $\frac{3}{2}n\log_2^2(d)$ or
$2n\log_2^2(d)$.

The Prover has to compute the Krylov space and the Krylov sequence plus the work
for $Z$, the work for $T$ and the recursive calls: 
$P(d)=(d\mu+2dn)+\PowerCertificate{d/2}+((d/2)\mu+2(d/2)n+P(d/2))$, 
so that the overall extra cost for the Prover is dominated by either 
$5d\mu+6dn$ or $7d\mu+6dn$.
For $d=2n$, the cost for the Prover without verification is
$\Seq(n)=2n\mu+4n^2$, which induces the last column of Table~\ref{tab:sc}.
\end{proof}

\section{Certificate for the determinant, the minimal and the characteristic
  polynomials}\label{sec:applis}
We denote by {\sc SeqCert} a certificate for Wiedemann's Krylov sequence. 
This can be for instance any of the subquadratic certificate of
Sections~\ref{sec:checkpoint},~\ref{sec:recursive} or~\ref{sec:nlogn}.

This induces directly a certificate for the minimal polynomial of a sequence: 
the Prover just produces the sequence, and the Verifier computes by himself the
minimal polynomial of the sequence via the fast extended Euclidean algorithm
(EEA). In a sufficiently large field, Wiedemann has shown that
this in turn induces a certificate for the minimal polynomial of a matrix, {\sc
  MinPoly}. In smaller fields one would need to use a certificate for a Block
Wiedemann sequence, and maybe some variants of the certificate of
Section~\ref{sec:blockwied}. 
Then a certificate for the determinant, {\sc Det}, is obtained via Wiedemann's
preconditioning, {\sc PreCondCyc}, insuring the square-freeness of the
characteristic polynomial. 
Finally, to get a certificate for the characteristic polynomial of a matrix,
{\sc CharPoly}, first ask for the characteristic polynomial, and then it is
sufficient to certify the determinant at a random point. 

We propose in Table~\ref{tab:applis} a summary of the reductions presented
in this section.
The details of these reductions and the proofs of the complexity claims shown in
Table~\ref{tab:applis} are given in
Theorems~\ref{thm:minpoly},~\ref{thm:det} and~\ref{thm:charpoly}. 

\begin{table}[htb]\center
\begin{tabular}{|c|c|}
\hline 
\multicolumn{2}{|c|}{{\sc MinPoly}}  \\
\hline 
Verifier & Verify({\sc SeqCert})+EEA \\
Communication & Communicate({\sc SeqCert})$+2n$ \\
Prover & Compute({\sc SeqCert})\\
\hline
\hline 
\multicolumn{2}{|c|}{{\sc Det}} \\
\hline
Verifier & Verifier({\sc MinPoly})\\
Communication & Communicate({\sc MinPoly}+{\sc PreCondCyc}) \\
Prover & Compute({\sc MinPoly}+{\sc
  PreCondCyc})\\
\hline
\hline
\multicolumn{2}{|c|}{{\sc CharPoly}} \\
\hline
Verifier & Verify({\sc Det})$+2n$ \\
Communication & Communicate({\sc Det})$+n$ \\
Prover & Compute({\sc CharPoly})+Compute({\sc Det})\\
\hline
\end{tabular}
\caption{Summary of the complexity reductions for the certification of the
  determinant, the minimal and the characteristic polynomials of sparse
  matrices}\label{tab:applis}
\end{table}

\subsection{{\sc MinPoly}}
\begin{theorem}[\cite{Wiedemann:1986:SSLE}]\label{thm:minpoly}
  Certifying the minimal polynomial can be reduced to the
  certification of Wiedemann's Krylov sequence.
\end{theorem}
\begin{proof}
The minimal polynomial of a linearly recurrent sequence can be computed by the
fast Euclidean algorithm, see,
e.g.,~\cite[Theorem~12.10]{VonzurGathen:2013:MCA}. Then Wiedemann's analysis
shows that in a sufficiently large field the minimal polynomial of a matrix can
be recovered by computing the lowest common multiple of the minimal polynomial
of sequences obtained by random
projections~\cite[Proposition~4]{Wiedemann:1986:SSLE}. 

Therefore, the work of the Prover is just that of computing minimal polynomials
of sequences at given vector projections. Communication is that of the two
vector projections, $2n$. Finally the work of the Verifier is to verify the
certificate for the sequence and then to apply the fast Euclidean algorithm, at
cost $n^{1+o(1)}$, to recover the minimal polynomial by himself.
\end{proof}

\subsection{{\sc Det}}
\begin{theorem}[\cite{Wiedemann:1986:SSLE}]\label{thm:det}
  Certifying the determinant can be reduced to the
  certification of the minimal polynomial.
\end{theorem}
\begin{proof}
We use the idea of~\cite[Theorem~2]{Wiedemann:1986:SSLE}: precondition the
initial matrix $A$ into a modified matrix $B$ whose characteristic polynomial is
square-free, and whose determinant is an easily computable modification of that
of $A$. For instance, such a {\sc PreCondCyc} preconditioner can be a diagonal matrix if the
field is sufficiently large~\cite[Theorem~4.2]{Chen:2002:EMP}
Precondition to get a square-free charpoly~\cite[Theorem~2]{Wiedemann:1986:SSLE}
and then certify the associated minpoly.
\end{proof}

\subsection{{\sc CharPoly}}
\begin{theorem}[\cite{jgd:2014:interactivecert}]\label{thm:charpoly}
  Certifying the characteristic polynomial can be reduced to the
  certification of the determinant.
\end{theorem}
\begin{proof}
The reduction is that of~\cite[Figure~1]{jgd:2014:interactivecert}:
the Prover computes the characteristic polynomial and sends it as a commitment
to the Verifier; 
then the Verifier gives a point $\lambda$ as challenge to the Prover which
responds with the determinant of $\lambda I_d-A$, and a certificate for that
determinant ($\lambda I_d-A$ remains sparse and costs no more than $\mu+n$ to be
applied to a vector).
Finally, the verifier simplify evaluates the commitment at $\lambda$ and checks
the equality with the certified determinant.
\end{proof}

\subsection{{\sc Det} over \Z}
Here the strategy is that of~\cite[\S 4.4]{jgd:2014:interactivecert}:
ask for {\sc MinPoly}, {\sc Det}, {\sc CharPoly} over~\Z.
After the commitment, the Verifier chooses a not so large prime, and ask for a
certificate of that same problem modulo the prime. 
Then the Verifier checks the certificate, and checks coherency with the integral
counterpart. 
On the one hand, the minimal and characteristic polynomial over {\sc \Z}
already occupy a quadratic space, so that taking modular images is already
quadratic.
On the other hand, for the determinant, this gives a linear time Verifier.

\bibliographystyle{plainurl} 
\bibliography{zkfhe} 

\begin{thebibliography}{10}

\bibitem{Bellare:1993:randomoracle}
Mihir Bellare and Phillip Rogaway.
\newblock Random oracles are practical: {A} paradigm for designing efficient
  protocols.
\newblock In Victoria Ashby, editor, {\em Proceedings of the 1st {ACM}
  Conference on Computer and Communications Security}, pages 62--73, Fairfax,
  Virginia, November 1993. ACM Press.
\newblock URL: \url{http://www-cse.ucsd.edu/users/mihir/papers/ro.pdf}.

\bibitem{Bernhard:2012:fiatshamir}
David Bernhard, Olivier Pereira, and Bogdan Warinschi.
\newblock How not to prove yourself: Pitfalls of the {Fiat-Shamir} heuristic
  and applications to helios.
\newblock In Xiaoyun Wang and Kazue Sako, editors, {\em Advances in Cryptology
  - {ASIACRYPT'12}}, volume 7658 of {\em Lecture Notes in Computer Science},
  pages 626--643. Springer, 2012.
\newblock URL:
  \url{http://www.uclouvain.be/crypto/services/download/publications.pdf.87e67d05ee05000b.6d61696e2e706466.pdf}.

\bibitem{Bostan:2003:Tellegen}
A.~Bostan, G.~Lecerf, and \'{E}. Schost.
\newblock Tellegen's principle into practice.
\newblock In {\em Proceedings of the 2003 International Symposium on Symbolic
  and Algebraic Computation}, ISSAC '03, pages 37--44, New York, NY, USA, 2003.
  ACM.
\newblock \href {http://dx.doi.org/10.1145/860854.860870}
  {\path{doi:10.1145/860854.860870}}.

\bibitem{Chen:2002:EMP}
Li~Chen, Wayne Eberly, Erich Kaltofen, B.~David Saunders, William~J. Turner,
  and Gilles Villard.
\newblock {Efficient matrix preconditioners for black box linear algebra}.
\newblock {\em Linear Algebra and its Applications}, 343-344:119--146, 2002.
\newblock \href {http://dx.doi.org/10.1016/S0024-3795(01)00472-4}
  {\path{doi:10.1016/S0024-3795(01)00472-4}}.

\bibitem{Chung:2010:delfhe}
Kai-Min Chung, Yael~Tauman Kalai, and Salil~P. Vadhan.
\newblock Improved delegation of computation using fully homomorphic
  encryption.
\newblock In Tal Rabin, editor, {\em Advances in Cryptology - {CRYPTO} 2010,
  30th Annual Cryptology Conference, Santa Barbara, {CA}, {USA}, August 15-19,
  2010. Proceedings}, volume 6223 of {\em Lecture Notes in Computer Science},
  pages 483--501. Springer, 2010.
\newblock \href {http://dx.doi.org/10.1007/978-3-642-14623-7_26}
  {\path{doi:10.1007/978-3-642-14623-7_26}}.

\bibitem{Coppersmith:1994:SHL}
Don Coppersmith.
\newblock Solving homogeneous linear equations over {$GF(2)$} via block
  {Wiedemann} algorithm.
\newblock {\em Mathematics of Computation}, 62(205):333--350, January 1994.
\newblock \href {http://dx.doi.org/10.2307/2153413}
  {\path{doi:10.2307/2153413}}.

\bibitem{Cramer:1997:Sigma}
Ronald John~Fitzgerald Cramer.
\newblock {\em Modular design of secure yet practical cryptographic protocols}.
\newblock PhD thesis, University of Amsterdam, 1996.

\bibitem{jgd:2014:interactivecert}
Jean-Guillaume Dumas and Erich Kaltofen.
\newblock Essentially optimal interactive certificates in linear algebra.
\newblock In Katsusuke Nabeshima, editor, {\em {ISSAC}'2014, Proceedings of the
  2014 ACM International Symposium on Symbolic and Algebraic Computation, Kobe,
  Japan}, pages 146--153. ACM Press, New York, July 2014.
\newblock \href {http://dx.doi.org/10.1145/2608628.2608644}
  {\path{doi:10.1145/2608628.2608644}}.

\bibitem{Fiat:1986:Shamir}
Amos Fiat and Adi Shamir.
\newblock How to prove yourself: Practical solutions to identification and
  signature problems.
\newblock In A.~M. Odlyzko, editor, {\em Advances in Cryptology - {CRYPTO'86}},
  volume 263 of {\em Lecture Notes in Computer Science}, pages 186--194.
  Springer-Verlag, 1987, 11--15~August 1986.
\newblock URL: \url{http://www.cs.rit.edu/~jjk8346/FiatShamir.pdf}.

\bibitem{Fiore:2012:PVD}
Dario Fiore and Rosario Gennaro.
\newblock Publicly verifiable delegation of large polynomials and matrix
  computations, with applications.
\newblock In {\em Proceedings of the 2012 ACM Conference on Computer and
  Communications Security}, CCS '12, pages 501--512, New York, NY, USA, 2012.
  ACM.
\newblock \href {http://dx.doi.org/10.1145/2382196.2382250}
  {\path{doi:10.1145/2382196.2382250}}.

\bibitem{Freivalds:1979:certif}
R{\=u}si{\c{n}}{\v{s}} Freivalds.
\newblock Fast probabilistic algorithms.
\newblock In J.~Be{\v{c}}v{\'a}{\v{r}}, editor, {\em Mathematical Foundations
  of Computer Science 1979}, volume~74 of {\em Lecture Notes in Computer
  Science}, pages 57--69, Olomouc, Czechoslovakia, September 1979.
  Springer-Verlag.
\newblock \href {http://dx.doi.org/10.1007/3-540-09526-8_5}
  {\path{doi:10.1007/3-540-09526-8_5}}.

\bibitem{Gentry:2014:nizkfhe}
Craig Gentry, Jens Groth, Yuval Ishai, Chris Peikert, Amit Sahai, and Adam
  Smith.
\newblock Using fully homomorphic hybrid encryption to minimize non-interative
  zero-knowledge proofs.
\newblock {\em Journal of Cryptology}, pages 1--24, 2014.
\newblock \href {http://dx.doi.org/10.1007/s00145-014-9184-y}
  {\path{doi:10.1007/s00145-014-9184-y}}.

\bibitem{Goldwasser:2008:delegating}
Shafi Goldwasser, Yael~Tauman Kalai, and Guy~N. Rothblum.
\newblock Delegating computation: interactive proofs for muggles.
\newblock In Cynthia Dwork, editor, {\em {STOC}'2008, Proceedings of the 40th
  Annual {ACM} Symposium on Theory of Computing, Victoria, British Columbia,
  Canada}, pages 113--122. ACM Press, May 2008.
\newblock \href {http://dx.doi.org/10.1145/1374376.1374396}
  {\path{doi:10.1145/1374376.1374396}}.

\bibitem{Kaltofen:1995:ACB}
Erich Kaltofen.
\newblock Analysis of {Coppersmith}'s block {Wiedemann} algorithm for the
  parallel solution of sparse linear systems.
\newblock {\em Mathematics of Computation}, 64(210):777--806, April 1995.
\newblock \href {http://dx.doi.org/10.2307/2153451}
  {\path{doi:10.2307/2153451}}.

\bibitem{KLYZ09}
Erich~L. Kaltofen, Bin Li, Zhengfeng Yang, and Lihong Zhi.
\newblock Exact certification in global polynomial optimization via
  sums-of-squares of rational functions with rational coefficients.
\newblock {\em Journal of Symbolic Computation}, 47(1):1--15, January 2012.
\newblock URL:
  \url{http://www.math.ncsu.edu/~kaltofen/bibliography/09/KLYZ09.pdf}, \href
  {http://dx.doi.org/10.1016/j.jsc.2011.08.002}
  {\path{doi:10.1016/j.jsc.2011.08.002}}.

\bibitem{Kaltofen:2011:quadcert}
Erich~L. Kaltofen, Michael Nehring, and B.~David Saunders.
\newblock Quadratic-time certificates in linear algebra.
\newblock In Anton Leykin, editor, {\em {ISSAC}'2011, Proceedings of the 2011
  ACM International Symposium on Symbolic and Algebraic Computation, San Jose,
  California, USA}, pages 171--176. ACM Press, New York, June 2011.
\newblock \href {http://dx.doi.org/10.1145/1993886.1993915}
  {\path{doi:10.1145/1993886.1993915}}.

\bibitem{Parno:2013:Pinocchio}
Bryan Parno, Jon Howell, Craig Gentry, and Mariana Raykova.
\newblock Pinocchio: Nearly practical verifiable computation.
\newblock In {\em Proceedings of the 2013 IEEE Symposium on Security and
  Privacy}, SP '13, pages 238--252, Washington, DC, USA, 2013. IEEE Computer
  Society.
\newblock \href {http://dx.doi.org/10.1109/SP.2013.47}
  {\path{doi:10.1109/SP.2013.47}}.

\bibitem{Thaler:2013:crypto}
Justin Thaler.
\newblock Time-optimal interactive proofs for circuit evaluation.
\newblock In Ran Canetti and JuanA. Garay, editors, {\em Advances in Cryptology
  - {CRYPTO'13}}, volume 8043 of {\em Lecture Notes in Computer Science}, pages
  71--89. Springer Berlin Heidelberg, 2013.
\newblock URL: \url{http://arxiv.org/abs/1304.3812}, \href
  {http://dx.doi.org/10.1007/978-3-642-40084-1_5}
  {\path{doi:10.1007/978-3-642-40084-1_5}}.

\bibitem{Villard:1997:further}
Gilles Villard.
\newblock Further analysis of {Coppersmith}'s block {Wiedemann} algorithm for
  the solution of sparse linear systems.
\newblock In Wolfgang~W. K{\"u}chlin, editor, {\em {ISSAC}'97, Proceedings of
  the 1997 ACM International Symposium on Symbolic and Algebraic Computation,
  Maui, Hawaii}, pages 32--39. ACM Press, New York, July 1997.
\newblock \href {http://dx.doi.org/10.1145/258726.258742}
  {\path{doi:10.1145/258726.258742}}.

\bibitem{VonzurGathen:2013:MCA}
Joachim von~zur Gathen and J{\"{u}}rgen Gerhard.
\newblock {\em Modern Computer Algebra {(3.} ed.)}.
\newblock Cambridge University Press, 2013.
\newblock \href {http://dx.doi.org/10.1017/CBO9781139856065}
  {\path{doi:10.1017/CBO9781139856065}}.

\bibitem{Wiedemann:1986:SSLE}
Douglas~H. Wiedemann.
\newblock Solving sparse linear equations over finite fields.
\newblock {\em IEEE Transactions on Information Theory}, 32(1):54--62, January
  1986.
\newblock \href {http://dx.doi.org/10.1109/TIT.1986.1057137}
  {\path{doi:10.1109/TIT.1986.1057137}}.

\end{thebibliography}
\end{document}